\documentclass[sigconf]{sig-alternate}

\usepackage{graphicx}
\usepackage{amsmath}
\usepackage{amsfonts}
\usepackage{latexsym}
\usepackage{amssymb}
\usepackage{amsthm}
\usepackage{color}
\usepackage{subfloat}
\usepackage{subfigure}
\usepackage{url}
\usepackage{mathrsfs}   
\usepackage[justification=raggedright]{caption}
\usepackage{pifont}
\usepackage{cite}

\usepackage{array}
\newcommand{\PreserveBackslash}[1]{\let\temp=\\#1\let\\=\temp}
\newcolumntype{C}[1]{>{\PreserveBackslash\centering}p{#1}}
\newcolumntype{R}[1]{>{\PreserveBackslash\raggedleft}p{#1}}
\newcolumntype{L}[1]{>{\PreserveBackslash\raggedright}p{#1}}
\usepackage{tabularx}

\usepackage[ruled]{algorithm2e}

\DeclareMathOperator*{\argmin}{arg\,min}

\newtheorem{theorem}{\bf Theorem}
\newtheorem{lemma}[theorem]{\bf Lemma}

\newtheorem{definition}{\bf Definition}

\newtheorem*{problemltpm}{\bf LPRC-LTPM Problem}
\newtheorem*{problemotpm}{\bf LPRC-OTPM Problem}
\newtheorem*{problemotcm}{\bf LPRC-OTCM Problem}

\newenvironment{remark}[1][Remark]{\begin{trivlist}
\item[\hskip \labelsep {\bfseries #1}]}{\end{trivlist}}

\setcopyright{rightsretained}

\acmPrice{15.00}

\begin{document}

\title{LEPA: Incentivizing Long-term Privacy-preserving Data Aggregation in Crowdsensing}

\author{
	\alignauthor
	Zhikun Zhang$^{\dagger}$,
	Shibo He$^{\dagger}$,
	Mengyuan Zhang$^{\dagger}$,
	Jiming Chen$^{\dagger}$, \\
	\affaddr{$^\dagger$ State Key Laboratory of Industrial Control Technology, Zhejiang University, China} \\
    \affaddr{\{zhangzhk,~s18he,~zhang418,~cjm\}@zju.edu.cn}
}

\maketitle \thispagestyle{empty}

\begin{abstract}
In this paper, we study the incentive mechanism design for real-time data aggregation, which holds a large spectrum of crowdsensing applications. Despite extensive studies on static incentive mechanisms, none of these are applicable to real-time data aggregation due to their incapability of maintaining PUs' long-term participation. We emphasize that, to maintain PUs' long-term participation, it is of significant importance to protect their privacy as well as to provide them a desirable \emph{cumulative compensation}. Thus motivated, in this paper, we propose LEPA, an efficient incentive mechanism to stimulate long-term participation in real-time data aggregation. Specifically, we allow PUs to preserve their privacy by reporting noisy data, the impact of which on the aggregation accuracy is quantified with proper privacy and accuracy measures. Then, we provide a framework that jointly optimizes the incentive schemes in different time slots to ensure desirable cumulative compensation for PUs and thereby prevent PUs from leaving the system halfway. Considering PUs' strategic behaviors and combinatorial nature of the sensing tasks, we propose a computationally efficient on-line auction with close-to-optimal performance in presence of NP-hardness of winner user selection. We further show that the proposed on-line auction satisfies desirable properties of truthfulness and individual rationality. The performance of LEPA is validated by both theoretical analysis and extensive simulations.

\end{abstract}



\section{introduction}
\subsection{Motivation}
The pervasive use of mobile devices (e.g., smartphones, smartwatchs, and tablet computers, etc.) embedded with multiple sensors (e.g., GPS, gyroscope and microphone, etc.) has given rise to a new sensing paradigm known as crowdsensing \cite{ganti2011mobile,he2017near,duan2017distributed}. Due to its low deployment cost and high sensing coverage, crowdsensing has spurred a wide range of applications including smart transportation, environmental monitoring and spectrum sensing, etc \cite{ganti2010greengps,cheng2014aircloud,jin2016dpsense}.



Real-time data aggregation holds a wild spectrum of crowdsensing applications, where Fusion Center (FC) continuously collects sensing data from Participatory Users (PUs) to provide real-time services. For instance, Waze \cite{waze}, a navigation app on smartphones that utilizes Waze users' (PUs) reported traffic data to estimate the real-time traffic condition. As shown in Figure \ref{fig:waze}, Waze users continuously report their speed and location information, as well as real-time traffic jams or accidents to Waze server (FC) through Waze app. Based on the information collected, the Waze server can carry out data aggregation and data analysis to provide real-time traffic information.
\begin{figure}[!ht]
\setlength{\abovecaptionskip}{0pt} 
\setlength{\belowcaptionskip}{0pt}
\begin{center}
\includegraphics[width=0.4\textwidth]{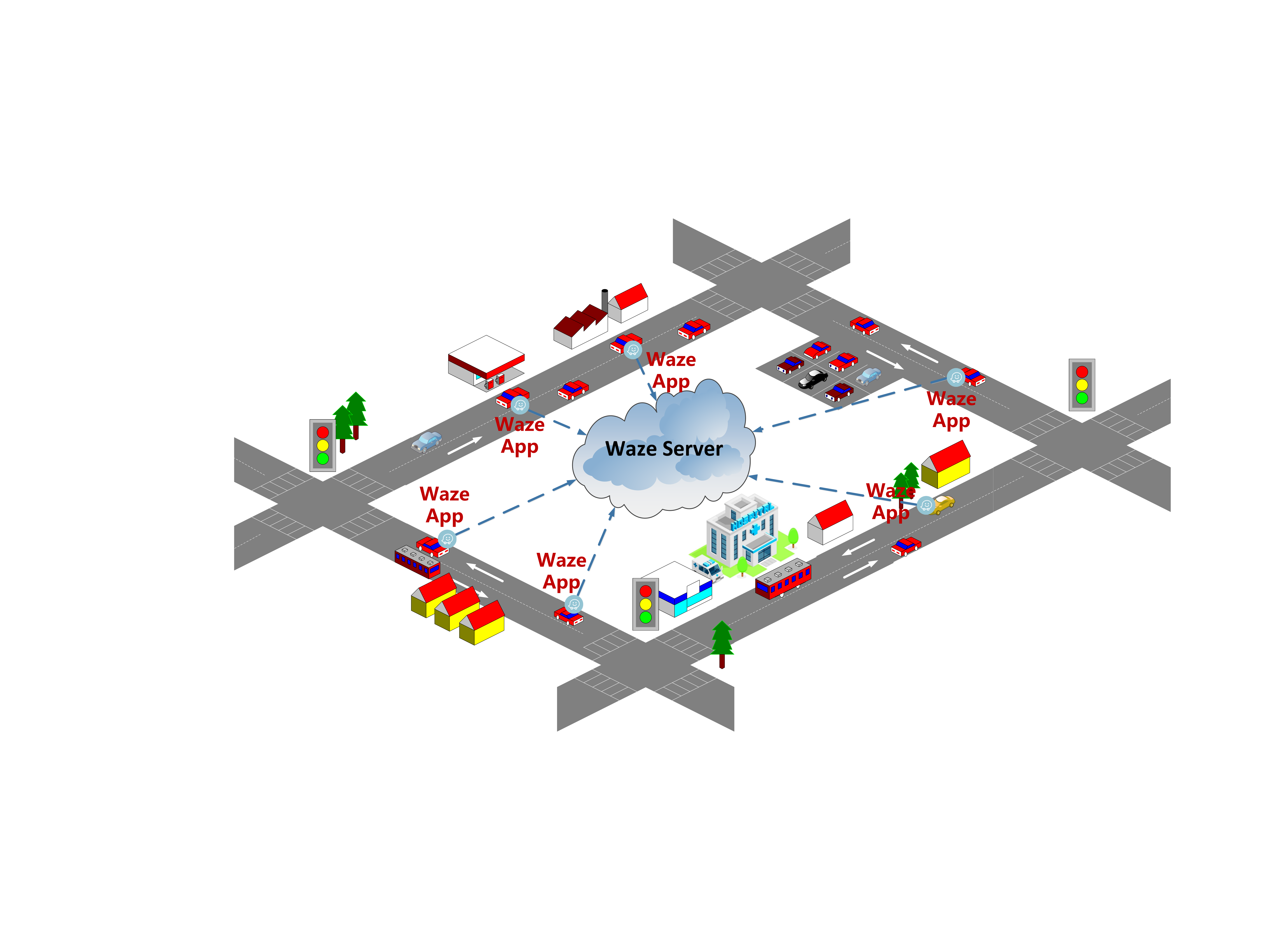}
\end{center}
\caption{An illustration of the Waze crowdsensing system.} \label{fig:waze}
\end{figure}

Clearly, participating in real-time data aggregation is costly for PUs, since it consumes not only PUs' time but also the system resources (e.g., computing and communication energy) in mobile devices. It is essential to design \emph{incentive mechanisms} to stimulate PUs' participation. In real-time data aggregation, PUs may drop out of the system halfway (e.g., close the account or even uninstall the crowdsensing app) if their long-term cumulative compensation expectations are not fulfilled, leading to insufficient number of PUs that continuously provide high quality services. Thus, we emphasize that the incentive mechanism design for real-time data aggregation needs to take a global perspective, guaranteeing that each PU has a desirable cumulative compensation and thus maintaining their long-term participation. This is quite different from previous studies on static incentive mechanism design, which encourages PUs to participate in one-time sensing tasks. 

To maintain PUs' long-term participation, there are two main concerns. On one hand, when contributing their sensing data, PUs suffer from potential privacy losses (e.g., the sensing data collected by the Waze server involves locational information), which hinders them from long-term participation. To dispel PUs' misgivings about potential privacy losses, an incentive mechanism should be well designed to allow PUs to fully control their own data privacy, and to provide adequate compensations for their privacy losses. On the other hand, participating in crowdsensing always incurs some indirect costs besides directly contributing sensing data. For example, when running in background, the crowdsensing app still consumes computation and bandwidth resources, even if the PU is not selected as a data contributor. However, the objective of the extensively studied static incentive mechanisms \cite{jin2016inception,yang2012crowdsourcing,jin2015quality,zhang2016privacy,luo2015crowdsourcing,duan2012incentive,peng2015pay,cheung2015distributed} is to select a subset of PUs that maximize FC's \emph{current utility} to carry out the sensing tasks, possibly rendering some PUs being unselected for a long period. Clearly, if a PU is rarely selected as a contributor in a long time, her cumulative cost can not be well compensated and loses interest in participation. Resultantly, incentive schemes in different time slots should be jointly optimized to maintain PUs' long-term participation.

In this paper, to resolve the above two concerns, we first allow PUs to preserve their data privacy by adding well-calibrated noise to their raw sensing data before reporting them. Such an approach bears two distinct merits: (i) PUs can fully control their private data; (ii) PUs' privacy-preserving levels (PPLs) can be quantified by the celebrated notion of \emph{differential privacy}. Then, we provide a framework that jointly optimizes the incentive schemes in different slots to maintain PUs' long-term participation. Specifically, we take an auction approach and propose LEPA\footnote{The name LEPA comes from \underline{L}ong-t\underline{E}rm \underline{P}rivacy-preserving data \underline{A}ggregation}, a novel incentive mechanism that maintains PUs' long-term participation in real-time data aggregation.

There are three major challenges in designing LEPA. Firstly, adding noise will inevitably impair the aggregation accuracy of FC. To achieve desirable aggregation accuracy in a cost-efficient manner, it is essential to quantify PUs' PPLs on FC's aggregation accuracy, which requires the calculation of the integral of the cumulative noises. This is difficult because for the mechanisms that achieve differential privacy, either the summation of $n$ Laplacian variables is a complex non-integrable distribution (i.e., Laplacian mechanism) or the probability density function of Gaussian distribution is non-integrable (i.e., Gaussian mechanism). Secondly, to maintain PUs' long-term participation, it is essential to guarantee that each PU will be selected at least once in several rounds, i.e., each PU's selected probability should be no less than a certain threshold. However, the calculation of such a selected probability depends on not only the past and current user selection strategies, but also future user selection strategies. This renders user selection strategies interrelated cross time domain, making the long-term incentive mechanism design more challenging. Thirdly, through adopting auction approach, FC selects PUs based on their biddings, which is determined by their sensing cost and privacy cost. Strategic PUs may misreport their biddings to maximize their own benefits, leading to high costs of achieving a desirable aggregation accuracy. It is challenging to devise a truthful incentive mechanism when PUs behave strategically.

\subsection{Summary of Main Contributions}
In this paper, we quantify the impact of PUs' PPLs on FC's aggregation accuracy with proper privacy and accuracy measures, and propose an on-line incentive mechanism to deal with the strategy coupling issue incurred by the long-term participation requirements. The contributions of this paper are summarized as follows:

\begin{itemize}
	\item Our work takes the first attempt to systematically study the long-term incentive mechanism for privacy-preserving data aggregation.
  	\item We allow PUs to add well-calibrated noise to their raw sensing data before reporting them. By creatively using tail bound theory, we successfully bypass the direct calculation of integrals of non-integrable distributions, and derive the quantitative impact of PUs' PPLs on FC's aggregation accuracy.
  	\item An on-line framework is proposed to deal with the strategy coupling issue by transforming the long-term participation requirement to a queue stability problem. Specifically, for each PU, we construct a virtual sensing request queue with a constant input rate, and its output rate is proportional to the PU's selected rate. To meet the long-term participation requirement is equivalent to ensuring the stability of all the virtual sensing request queues. We then jointly optimize FC's total payment and the queue length.
  	\item Considering PUs' strategic behaviors and the combinatorial nature of the tasks, we design an incentive mechanism based on reverse combinatorial auction, where FC acts as an auctioneer and purchases private sensing data from PUs. Due to the NP-hardness of the combinatorial auction, we propose a computationally efficient mechanism with close-to-optimal performance, meanwhile guaranteeing individual rationality and truthfulness.
  	\item Both theoretical analysis and extensive simulations are conducted to corroborate the performance of the proposed incentive mechanism.
\end{itemize}

\subsection{Organization of This Paper}
The remaining of this paper is organized as follows. In Section \ref{relatedwork}, we discuss the related work. We introduce the problem formulation in Section \ref{preliminary} and provide the design details of LEPA in Section \ref{mechanismdesign}. Section \ref{theoreticalanalysis} provides theoretical analysis to LEPA, and Section \ref{simulation} conducts extensive simulations to illustrate the effectiveness of LEPA. Section \ref{conclusion} concludes this paper.

\section{Related Work}\label{relatedwork}
Game theory has been widely adopted to capture PUs' strategic behaviors for incentive mechanism design in the crowdsensing system. These incentive mechanisms are based on either auction \cite{jin2016inception,yang2012crowdsourcing,jin2015quality,zhang2016privacy} or other game-theoretical models \cite{luo2015crowdsourcing,duan2012incentive,peng2015pay,cheung2015distributed}. Most of the previous studies consider static incentive, which encourages PUs to participate in one-time sensing tasks. These works are inapplicable to the real-time crowdsensing applications. Recently, a few incentive mechanisms have been proposed to stimulate PUs' long-term participation. In \cite{lee2010sell}, Lee et al. proposed to provide \emph{virtual credit} to PUs who lost in the previous auction round, so that their winning probabilities increase in the future rounds. Gao et al. in \cite{gao2015providing} proposed a long-term sensor selection mechanism for a general location-aware crowdsensing system. However, both of them study an abstract and general crowdsensing system without addressing PUs' privacy concern. Further, they fail to provide a truthful mechanism with proofed performance guarantee.



Another line of related works view privacy as a good and aim to compensate PUs' privacy losses. In their seminal work \cite{ghosh2015selling}, Ghosh et al. designed an auction mechanism to maximize FC's utility subject to a budget constraint. Thereafter, several improved mechanisms have been proposed, especially considering the correlation between the privacy preference and the private data. Nevertheless, the existing works either deprived PUs' direct control on their data privacy (since FC is assumed to be trustworthy and responsible for protecting their privacy) \cite{jin2016inception,fleischer2012approximately,ligett2012take,nissim2014redrawing}, or focus on PUs' equilibrium behavior \cite{wang2016value}, which may end up with an inefficient equilibrium, i.e., FC may not achieve a desirable aggregation accuracy in a cost-efficient manner.


Different from the existing works on abstract crowdsensing tasks, our study investigates a kind of widely deployed crowdsensing tasks, i.e., real-time data aggregation. We proposed a novel paradigm that allows PUs to add well-calibrated noise to their raw sensing data before reporting them, enabling PUs' fully control of their data privacy. Then, rather than incentivize PUs' one-time participation, our incentive mechanism can jointly optimize FC's total payment and PUs' long-term participation requirements, meanwhile satisfies proofed truthfulness and near-to-optimal performance.

\section{Problem Formulation}\label{preliminary}
In this section, we present the system overview, as well as descriptions to the quantitative impact of PUs' PPLs on FC's aggregation accuracy, auction model and design objectives.

\subsection{System overview}
The crowdsensing system considered in this paper consists of an FC and a set $\mathcal{U} = \{u_1, u_2, \ldots, u_n\}$ of $n$ PUs, as illustrated in Figure \ref{fig:systemmodel}. FC has a task pool which possesses a set $\mathcal{T} = \{\tau_1, \tau_2, \ldots, \tau_k\}$ of $k$ sensing tasks. The sensing tasks request data slot by slot, where each time slot ranges from several minutes to several hours, depending on the updating frequency requirements of the sensing tasks. At time slot $t$, each PU can execute a subset of tasks in the task pool based on the capability of her mobile device and her current location. Let $y_{ij}(t) \in \{0,1\}$ denote whether PU $i$ can execute task $j$, and $y_i(t)$ denote PU $i$'s capability vector in time slot $t$. For each task $\tau_j$, FC should collect plenty of sensing data from PUs and carry out some aggregation operations, such as average or histogram, to abstract some valuable patterns. For easy exposition, in this paper, we will investigate the average aggregation\footnote{We leave the discussion of other kinds of data aggregations in future work.}, which constitute a large portion of currently deployed crowdsensing systems, such as traffic monitoring systems that leverage vehicular PUs' GPS data to estimate average traffic speed of a specific road. Then FC can select a subset of PUs to implement all sensing tasks. We denote whether PU $i$ is selected in time slot $t$ as $x_i(t) \in \{0,1\}$.

Clearly, the sensing data usually contain sensitive information about PUs, which will hinder them from providing their raw sensing data. Different from most of the previous studies on privacy-preserving data aggregation, we do not assume FC to be trustworthy since it maybe compromised or the communication channel maybe eavesdropped. Making a paradigm shift, we remove the trustworthy FC assumption and allow PUs to add well-calibrated noises to their raw sensing data before reporting them. Specifically, the workflow of the proposed crowdsensing system is as follows.



\begin{figure}[!ht]
\setlength{\abovecaptionskip}{0pt} 
\setlength{\belowcaptionskip}{0pt}
\begin{center}
\includegraphics[width=0.4\textwidth]{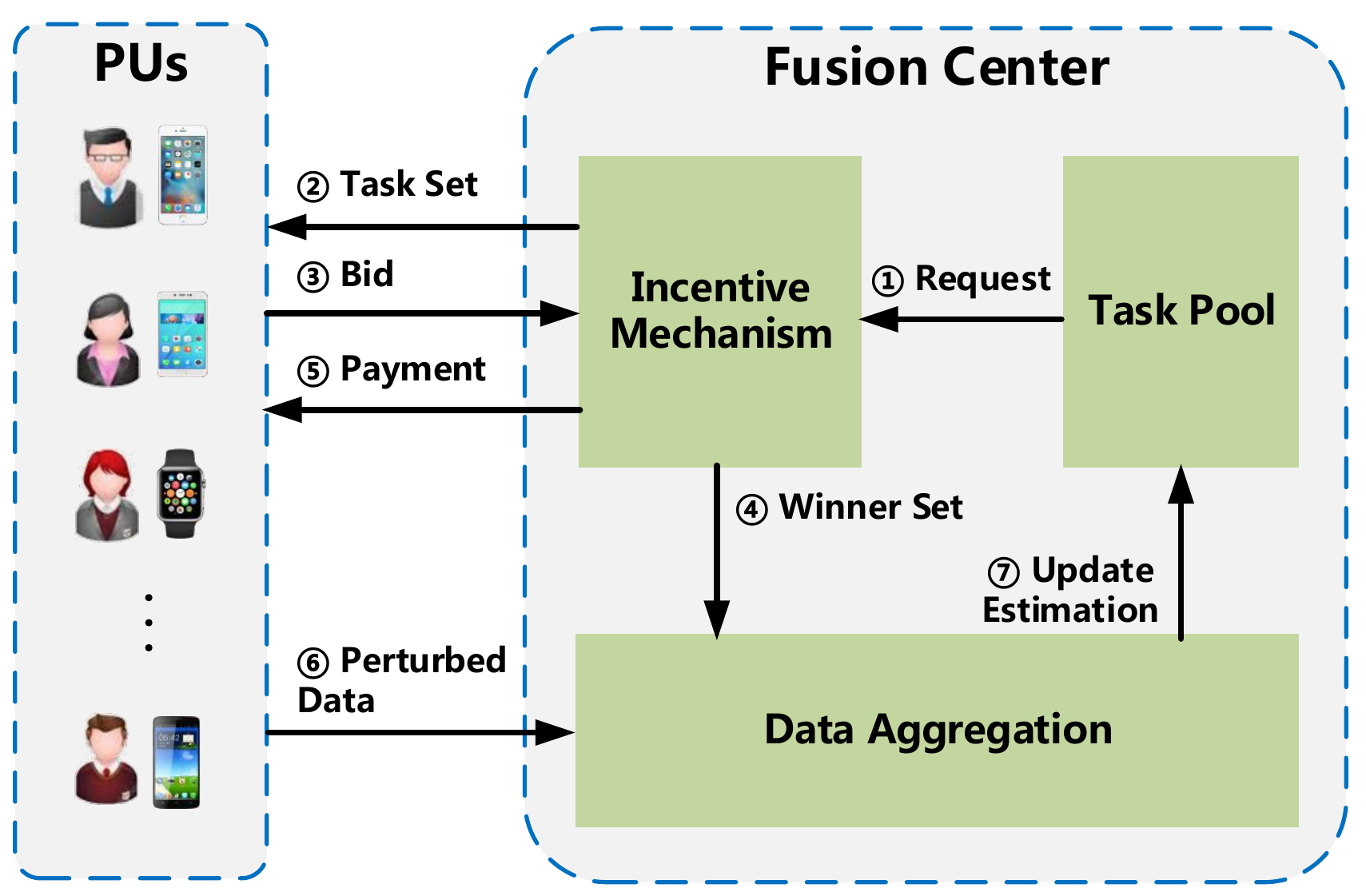}
\end{center}
\caption{A framework of the proposed crowdsensing system. The circled numbers stand for procedures of the system operation.} \label{fig:systemmodel}
\end{figure}


\begin{itemize}
  \item Firstly, at time slot $t$, each task launches a sensing request, and specifies its aggregation accuracy requirement to the incentive mechanism (step \ding{172}).
  \item \textbf{Incentive mechanism:} Upon receiving the sensing requests, FC conducts a reverse combinatorial auction to purchase PUs' private sensing data. Specifically, FC first broadcasts all sensing requests and specifies a homogeneous PPL to PUs (step \ding{173}). Then, PU $i$ submits a bidding $b_i(t) = \{ \Gamma_i(t), b_i^s(t), b_i^p(t) \}$ to FC, where $\Gamma_i(t)$ is the set representation of $y_i(t)$\footnote{We will use $\Gamma_i(t)$ and $y_i(t)$ interchangeably as convenient in the following parts.}, $b_i^s(t)$ and $b_i^p(t)$ are the sensing cost and unit privacy cost of PU $i$ at time slot $t$ (step \ding{174}). Finally, FC selects an optimal set of winners $\mathcal{S} = \{ u_1, u_2, \ldots, u_l \}$ (step \ding{175}) that could implement all sensing tasks, and provides them corresponding payments $\mathcal{P} = \{ p_1, p_2, \ldots, p_l \}$ (step \ding{176}).
  \item \textbf{Data aggregation:} Each selected PU adds a well-calibrated noise to her raw sensing data, i.e., $\tilde{d_i} = d_i + \eta_i$, and uploads the perturbed sensing data $\tilde{d_i}$ to FC (step \ding{177}). FC then carries out data aggregation for each task.
  \item Finally, the aggregation results are leveraged to update the estimation of sensing tasks (step \ding{178}).
\end{itemize}

\begin{remark}
In practice, once the auction is finished, the crowdsensing apps of winner PUs will collect the sensing data and add the specified levels of noise before reporting them to FC automatically, such that PUs can not cheat on their noise distribution, as well as the corresponding noise level, e.g., adding higher noise level to better preserve privacy.
\end{remark}

\subsection{Privacy versus Accuracy}
Intuitively, adding noise will inevitably deteriorate FC's aggregation accuracy. In order to provide accurate incentive in presence of noisy data, it is essential to quantify the impact of the noise level, i.e., PUs' PPLs, on FC's aggregation accuracy with proper privacy and accuracy definition.

We adopt the celebrated notion of differential privacy \cite{dwork2006calibrating}, which has rigorous mathematical definition and is independent of attackers' background knowledge, to define PUs' privacy. Intuitively, differential privacy guarantees that the attackers can not distinguish between the neighboring inputs with high confidence given a randomized output. Traditional differential privacy assumes a trustworthy curator which is not applicable to our system. In this paper, we concern about local differential privacy where PUs add noises to their raw data locally. 

Defining local differential privacy requires the following adjacency definition:
\begin{definition}[$\zeta_i$-adjacency]
Two real numbers $d$ and $d'$ are $\zeta_i$-adjacency if $|d - d'| \leq \zeta_i$, where $\zeta_i$ is the length of the value range of $d$.
\end{definition}

Based on the above adjacency definition, we define the local differential privacy as follows:
\begin{definition}[$\epsilon_i$-local differential privacy]\label{definition:dp}
A random algorithm $\{\mathcal{A}:R \rightarrow R|\mathcal{A}(d_i)=d_i+\eta_i\}$ achieves $\epsilon_i$-local differential privacy, if for all pairs of $\zeta_i$-adjacency data $d_i$ and $d_i'$, and observation $d^{obs}$,
\begin{align*}
Pr[\mathcal{A}(d_i)=d^{obs}]\leq e^{\epsilon_i}Pr[\mathcal{A}(d_i')=d^{obs}].
\end{align*}
\end{definition}

Intuitively, smaller $\epsilon_i$ means higher PPL, since it is more difficult to distinguish $d_i$ from $d_i'$ given the observation $d^{obs}$. We then adopt the widely used Laplacian mechanism to achieve local differential privacy. It is known in \cite{dwork2014algorithmic} that, when the Laplacian mechanism is used, i.e., $\eta_i \sim Lap(0,a_i)$, we can achieve $\epsilon_i$-local differential privacy if $a_i = \frac{\zeta_i}{\epsilon_i}$.


Then, we define the aggregation accuracy in Definition \ref{definition:accuracy}.
\begin{definition}[$(\alpha,\delta)$-accuracy]\label{definition:accuracy}
The aggregation $\hat{s}$ of privacy-preserving sensing data is said to achieve $(\alpha,\delta)$-accuracy if
\begin{align*}
Pr[|s - \hat{s}| \geq \alpha] \leq \delta.
\end{align*}
\end{definition}

From statistical perspective, $\alpha$ is the confidence interval of the aggregation error and $\delta$ is the confidence level.

With the privacy and accuracy definitions, we derive the quantitative impact of PUs' PPLs on FC's aggregation accuracy in the following lemma.
\begin{lemma}\label{lemma:relationship}
In time slot $t$, task $\tau_j$ achieves $(\alpha_j,\delta_j)$-accuracy if
\begin{small}
\begin{align}\label{formula:relationship}
\frac{\sum\limits_{i:u_i \in \mathcal{U}} x_i(t) y_{ij}(t) \frac{\zeta}{\epsilon_i^2(t)}}
{\left(\sum\limits_{i:u_i \in \mathcal{U}} x_i(t) y_{ij}(t)\right)^2}
\leq \frac{\alpha_j^2(t) \delta_j(t)}{2},
\end{align}
\end{small}
where $\zeta$ is the range of all PUs' sensing data. 
\end{lemma}

The proof can be found in Appendix \ref{appendix:relationship}. In this paper, we consider the scenario where FC broadcasts a homogeneous PPL $\epsilon$ to all PUs\footnote{The selection principle of $\epsilon$ will be discussed in the simulation part.}, i.e., all PUs share the same $\epsilon$. To simplify the analysis, we transform (\ref{formula:relationship}) to
\begin{align*}
\sum\limits_{i: u_i \in \mathcal{U}} x_i(t) y_{ij}(t) \geq \frac{2\zeta}{\epsilon^2 \alpha_j^2(t) \delta_j(t)}.
\end{align*}

We denote $r_j(t) = \frac{2\zeta}{\epsilon^2 \alpha_j^2(t) \delta_j(t)}$ as task $\tau_j$'s \emph{aggregation accuracy requirement}, and leverage a vector $\textbf{r(t)} \\ = (r_1(t),\cdots,r_m(t))$ to denote the accuracy requirements of all tasks.

\subsection{Auction Model}
In this paper, we investigate the long-term incentive mechanism for privacy-preserving data aggregation, where tasks require real-time sensing data from PUs. We assume that PUs are all selfish and strategic, aiming to maximize their own utility. Thus, we model our incentive mechanism as the Long-term Privacy-preserving Reverse Combinatorial (LPRC) auction.
\begin{definition}[LPRC Auction] \label{definition:auction}
In a long-term privacy-preserving reverse combinatorial auction, each PU $u_i$ is interested in a subset of tasks in time slot $t$, i.e., $y_i(t)$, and can bid her sensing cost $b_i^s(t)$ and unit privacy cost $b_i^p(t)$ for reporting the sensing data. Both $b_i^s(t)$ and $b_i^p(t)$ are $u_i$'s private information.
\end{definition}

Then, we define PU's utility and FC's payment in Definition \ref{definition:pu_utility} and Definition \ref{definition:fc_payment}.

\begin{definition}[PU's Utility]\label{definition:pu_utility}
The utility of PU $i$ is given by
\begin{align*}
U_i(t) = \left\{\begin{array}{ll} p_i(t) - c_i^s(t) - c_i^p(t)\epsilon, &i\in \mathcal{S}, \\
0, &otherwise, \end{array}\right.
\end{align*}
\end{definition}

\begin{definition}[FC's Payment]\label{definition:fc_payment}
In time slot $t$, FC's total payment $P(t)$ is given by
\begin{align*}
P(t) = \sum\limits_{i:u_i\in\mathcal{S}}p_i(t)
\end{align*}
\end{definition}

\subsection{Design Objective}
Since PUs are strategic and their sensing costs as well as privacy costs are unknown to FC, they may misreport their truthful biddings to achieve higher profits. For example, a selfish PU may report a higher sensing cost to achieve more compensation. Therefore, we should guarantee that our LPRC auction satisfies the following truthfulness objectives.
\begin{definition}[Truthfulness]
A LPRC auction satisfies truthfulness if and only if submitting the truthful bidding $b_i = (\Gamma_i, c_i^s, c_i^p)$ is the dominant strategy for all PUs, i.e., $U_i(b_i, b_{-i}) \\ \geq U_i(b_i', b_{-i}), \forall b_i' \neq b_i$.
\end{definition}

Truthfulness guarantees that submitting the truthful bidding can maximize each PU's utility, so that all PUs have no incentives to misreport their biddings. Apart from truthfulness, our mechanism should also satisfies individual rationality, i.e., each PU's utility is non-negative so that all PUs are willing to participate.
\begin{definition}[Individual Rationality]
A LPRC auction satisfies individual rationality if and only if $U_i \geq 0$ for all $u_i \in \mathcal{U}$.
\end{definition}

For the real-time data aggregation tasks, PUs suffer from some indirect cost when running crowdsensing app in background. If a PU is rarely selected for a long time, she will probably choose to leave the platform. To maintain PUs' activity, we should guarantee that the probability that each PU being selected is no less than the minimum selected probability requirement $D$. Formally, we define this requirement as the long-term participation constraint in Definition \ref{definition:longterm}.
\begin{definition}[Long-term Participation] \label{definition:longterm}
A LPRC auction satisfies long-term participation constraints if and only if $u_i$'s selected probability satisfies
\begin{align*}
\frac{1}{T} \sum_{t \in T}x_i(t) \geq D, \ \ \ \forall u_i \in \mathcal{U}.
\end{align*}
\end{definition}

In practice, we can conduct questionnaires to determine the minimum threshold $D$ that maintains PUs' long-term participation.

\section{Long-term Incentive Mechanism}\label{mechanismdesign}
In this section, we first provide the design details on LEPA, including mathematical formulation, on-line auction transformation and on-line auction design. Then, we discuss some practical issues in implementing LEPA.

\subsection{Mathematical Formulation}
As aforementioned, LEPA is based on the LPRC auction defined in Definition \ref{definition:auction}. In this paper, we aim to design an LPRC auction that minimizes FC's long-term total payment while guaranteeing the aggregation accuracy requirements of all sensing tasks. Therefore, in each time slot, FC determines the winner PUs and the corresponding payments by solving the following LPRC long-term total payment minimization (LPRC-LTPM) problem.
\begin{problemltpm}\label{problem:ltpm}
\begin{small}
\begin{align}
\min \ \ \ \ &\sum\limits_{t\in T}\sum\limits_{i:u_i\in\mathcal{U}}p_i(t)x_i(t) \\
s.t. \ \ \ \
&\frac{1}{T}\sum\limits_{t\in T}x_i(t) \geq D, \forall u_i \in \mathcal{U}, 
\label{constraint:longterm} \\
&\sum\limits_{i: u_i \in \mathcal{U}} x_i(t) y_{ij}(t) \geq r_j(t), \forall \tau_j \in \mathcal{T},  
\label{constraint:accuracy} \\
&x_i(t) \in \{0,1\}, p_i(t) \in [0,+\infty],
\end{align}
\end{small}
\end{problemltpm}

\textbf{Constants:} In each time slot, the LPRC-LTPM problem takes as input the task set $\mathcal{T}$, PU set $\mathcal{U}$, accuracy requirements vector $\textbf{r(t)}$, bidding profile $(y_i(t), b_i^s(t), b_i^p(t)), \forall u_i \in \mathcal{U}$, as well as the minimum selected probability requirement $D$.

\textbf{Variables:} The LPRC-LTPM problem has a vector of $n$ binary variables $\textbf{x(t)} = (x_1(t), x_2(t), \cdots, x_n(t))$. Any $x_i(t) = 1$ indicates that $u_i$ is selected in time slot $t$ (i.e., $u_i \in \mathcal{S}$), whereas $x_i(t) = 0$ means $u_i \not\in \mathcal{S}$. Another vector of $n$ variables $\textbf{p(t)} = (p_1(t), p_2(t), \cdots, p_n(t))$ is the payment profile where $p_i(t)$ takes non-negative real value. If $u_i \not\in \mathcal{S}$ in time slot $t$, $p_i(t) = 0$.

\textbf{Objective function:} The objective function is set as the long-term total payment.

\textbf{Constraints:} Constraint (\ref{constraint:longterm}) is the long-term participation constraint as defined in Definition \ref{definition:longterm}. Constraint (\ref{constraint:accuracy}) guarantees that each task $\tau_j$'s accuracy requirement is satisfied. In addition, any solution to LPRC-LTPM problem should satisfy two inherent constraints, i.e., truthfulness and individual rationality.

Notice that constraint (\ref{constraint:longterm}) is a time average constraint, where current strategy is coupled with future strategies. However, PUs' biddings in the future are unknown, making the user selection in current time slot challenging. To tackle such a challenge, we transform the LPRC-LTPM problem to an on-line auction design problem.

\subsection{On-line Auction Transformation}

The main idea of the on-line auction is to transform the long-term participation requirements to the queue stability requirements, and leverage Lyapunov optimization theorem \cite{neely2010stochastic} to jointly optimize the queue stability and FC's total payment. Following this idea, we first construct a virtual request queue with arrival rate $D$ for each PU to buffer her sensing request. Clearly, the backlog of a virtual request queue corresponds to the cumulative sensing requests of each PU. Then, if $u_i$ is selected at time slot $t$, i.e., $x_i(t)=1$, one virtual request leaves the backlog.

Based on the above virtual request queue definition, we have the following queue dynamics:
\begin{align*}
q_i(t+1) = [q_i(t)-x_i(t)]^{+} + D,
\end{align*}
where $q_i(t)$ denotes the virtual request queue backlog of $u_i$ at time slot $t$, and $[x]^+ = \max\{x,0\}$.

We can analyze the queue stability by the Lyapunov drift, which is defined as the difference between the Lyapunov functions in two adjacent time slots. Normally, the Lyapunov function is defined as quadratic sum of all queue backlogs, i.e.,
\begin{align*}
L(t) \triangleq \frac{1}{2}\sum\limits_{i: u_i \in \mathcal{U}} (q_i(t))^2.
\end{align*}

Thus, the Lyapunov drift is given by
\begin{align*}
\Delta(t) \triangleq L(t+1)-L(t).
\end{align*}

The Lyapunov drift theorem indicates that an algorithm greedily minimizing $\Delta(t)$ in each time slot guarantees the stabilities of all queues. Once all virtual request queues are stabilized, the long-term participation requirements are automatically satisfied.

Recall that the objective of LPRC-LTPM problem is to minimize FC's total payment meanwhile satisfying all long-term participation constraints, i.e., stabilizing all the virtual request queues. By the Lyapunov optimization theorem, we can minimize the following drift-plus-penalty to jointly stabilize the queues and optimize the objection function.
\begin{align*}
\Delta(t) + \gamma\sum\limits_{i:u_i\in\mathcal{U}}p_i(t)x_i(t),
\end{align*}
where $\gamma$ is a tuning parameter used to achieve the desirable tradeoff between queue stability and objective optimality.

Notice that $\Delta(t)$ is a quadratic function which is difficult to handle in the optimization problem. Thus, we focus on minimizing a specific linear upper bound of the drift-plus-penalty given by
\begin{small}
\begin{align}
\Delta(t) = & L(t+1) - L(t) \\
= & \frac{1}{2}\left(\sum\limits_{i:u_i\in\mathcal{U}}([q_i(t) - x_i(t)]^+ + D)^2
- \sum\limits_{i:u_i\in\mathcal{U}} q_i(t)^2\right) \\
\leq & \frac{1}{2} \Bigg(\sum\limits_{i:u_i\in\mathcal{U}}\big(q_i(t)^2 + D^2 + x_i(t)^2 + 2q_i(t)(D - x_i(t))\big)
\\ 
 & - \sum\limits_{i:u_i\in\mathcal{U}} q_i(t)^2\Bigg) \\
\leq & \sum\limits_{i:u_i\in\mathcal{U}}\frac{D^2 + 1}{2} + \sum\limits_{i:u_i\in\mathcal{U}} q_i(t)D - \sum\limits_{i:u_i\in\mathcal{U}} q_i(t)x_i(t), \label{formula:upperbound}
\end{align}
\end{small}
where the first inequality holds since $(\max[q-x,0]+D)^2 \leq q^2 + D^2 + x^2 + 2q(D-x)$, and the second inequality holds because $x_i(t) \in \{0,1\}$.

Notice that the first two parts of (\ref{formula:upperbound}) are constants. Minimizing (\ref{formula:upperbound}) is equivalent to minimizing \\ $\sum\limits_{i:u_i\in\mathcal{U}} -q_i(t)x_i(t))$. Therefore, we can transform the LPRC-LTPM problem to the following LPRC on-line total payment minimization (LPRC-OTPM) problem:
\begin{problemotpm}\label{problem:otpm}
\begin{align*}
\min \ \ \ \ &\sum_{i: u_i \in \mathcal{U}}[\gamma p_i(t) - q_i(t)]x_i(t) \\
s.t. \ \ \ \
&\sum\limits_{i: u_i \in \mathcal{U}} x_i(t) y_{ij}(t) \geq r_j(t), \forall \tau_j \in \mathcal{T} \\
&x_i(t) \in \{0,1\}, p_i(t) \in [0,+\infty],
\end{align*}
\end{problemotpm}

It is easy to show that the LPRC-OTPM problem is polynomial time reducible to the minimum weight set cover problem, meaning the LPRC-OTPM problem is NP-hard. It follows directly that calculating the optimal winner PUs and the corresponding payments profile are computationally inefficient when the number of PUs and tasks become large. To tackle such a problem, we present an on-line LPRC auction with near-to-optimal performance in the following subsection.



\subsection{On-line LPRC Auction Design}
\begin{algorithm}[htb]
\SetCommentSty{small}
\LinesNumbered

\caption{On-line LPRC Auction Design}
\label{algorithm:lyapunov}

\KwIn{$\epsilon$, $\textbf{b}$, $\textbf{r}$, $\gamma$, $\mathcal{U}$, $\mathcal{T}$;}
\KwOut{$\mathcal{S}$, $\textbf{p}$;}
\textbf{Initialization:} $\textbf{q} = \{0,\cdots,0\}$; \\
\ForEach{time slot $t = 0,1,\cdots,T$}
{
    Selection Rule: \\ 
    \ \ \ \ run \textbf{Algorithm \ref{algorithm:selection}} and output the winner user set $\mathcal{S}$; \\
    Payment Rule: \\
    \ \ \ \ run \textbf{Algorithm \ref{algorithm:price}} and output the payment profile $\textbf{p}$; \\
    Updating Rule: \\
    \ \ \ \ $q_i(t+1) = [q_i(t)-x_i(t)]^{+} + D_n$;
}

\end{algorithm}
We illustrate the on-line LPRC auction in Algorithm \ref{algorithm:lyapunov}. Notice that the on-line LPRC auction focus on user selection and payment determination strategy in time slot $t$. For ease of exposition, we remove the time index $t$ for related variables, e.g., we simplify the payment profile $\textbf{p(t)}$ to $\textbf{p}$. In the beginning of each time slot, the on-line LPRC auction runs Algorithms \ref{algorithm:selection} and \ref{algorithm:price} to determine the winner user set $\mathcal{S}$ and the corresponding payment profile $\textbf{p}$. Then, updating each $u_i$'s virtual request queue as line 7. We depict the winner user selection algorithm and payment determination algorithm in the following parts.

\begin{algorithm}[th]
\SetCommentSty{small}
\LinesNumbered

\caption{Winner User Selection}
\label{algorithm:selection}

\KwIn{$\epsilon$, \textbf{b}, \textbf{r}, $\gamma$, $\textbf{q}$, $\mathcal{U}$, $\mathcal{T}$;}
\KwOut{$\mathcal{S}$;}

\textbf{Initialization:}
$\mathcal{S} \leftarrow \varnothing$; 
$\textbf{r'} \leftarrow \textbf{r}$;

\While{$\sum_{j:\tau_j \in \mathcal{T}} r'_j \neq 0$}
{
    \tcp{Select the PU with minimum bidding accuracy ratio.}
    
    $l = \argmin_{i \in \mathcal{U}}\frac{b_i^s + b_i^p\epsilon - \frac{q_i}{\gamma}}{\sum_{j:\tau_j \in \Gamma_i} \min\{r'_j, 1\}}$;
    
    $\mathcal{S} \leftarrow \mathcal{S} \cup \{u_l\}$;
    
    $\mathcal{U} \leftarrow \mathcal{U} \setminus \{u_l\}$;
    
    \tcp{Update $r'_j$}
    
    \ForEach{$j: \tau_j \in \Gamma_l$}
    {
        $r'_j \leftarrow r'_j - \min\{r'_j, 1\}$;
    }
}

\Return{$\mathcal{S}$};
\end{algorithm}
The winner user selection algorithm given in Algorithm \ref{algorithm:selection} takes as input the PPL $\epsilon$, bidding profile $\textbf{b}$ (including each $u_i$'s cost $c_i^s,c_i^p$ and capability vector $y_i$), accuracy requirement vector $\textbf{r}$, tuning parameter $\gamma$, queue vector $\textbf{q}$, the PU set $\mathcal{U}$ and the task set $\mathcal{T}$. Firstly, we initialize the winner user set $\mathcal{S}$ (line 1). Then, in each loop, the PU with minimum bidding accuracy ratio is selected until all tasks' accuracy requirements are satisfied, i.e., $\sum_{j:\tau_j \in \mathcal{T}} r_j = 0$. The bidding accuracy ratio is defined in line 3, meaning each PU's contribution to the optimization problem. Finally, in the end of each loop, the remaining accuracy requirement of each task is updated (line 7).

\begin{algorithm}[ht]
\SetCommentSty{small}
\LinesNumbered

\caption{Payment Determination}
\label{algorithm:price}

\KwIn{$\epsilon$, $\textbf{b}$, $\textbf{r}$, $\gamma$, $\textbf{q}$, $\mathcal{U}$, $\mathcal{T}$, $\mathcal{S}$;}
\KwOut{$\textbf{p}$}

\textbf{Initialization:}
$\textbf{p} \leftarrow (0,\cdots,0)$;

\ForEach{$i:u_i \in \mathcal{S}$}
{
    run Algorithm \ref{algorithm:selection} on $\mathcal{U} \setminus \{u_i\}$;
    
    $\mathcal{S'} \leftarrow$ the winner set of step 3;
    
    \tcp{Calculate payment}
    
    \ForEach{$k:u_k \in \mathcal{S'}$}
    {
        $r'_j \leftarrow$ task $\tau_j$'s $r'_j$ when $u_k$ was selected;
        
        $p_i \leftarrow \max\{p_i, \frac{\sum\limits_{j:\tau_j \in \Gamma_i} \min\{r'_j, 1\}}{\sum\limits_{j:\tau_j \in \Gamma_k} \min\{r'_j, 1\}} (b_k^s + b_k^p\epsilon - \frac{q_k}{\gamma}) + \frac{q_i}{\gamma}\}$;
    }
}

\Return{$\textbf{p}$}
\end{algorithm}
Next, we decide the payments to the winner PUs by Algorithm \ref{algorithm:price}. The payment determination algorithm takes as input all the input parameters of Algorithm \ref{algorithm:selection}, as well as the winner user set $\mathcal{S}$ returned by Algorithm \ref{algorithm:selection}. Firstly, we initialize the payment vector $\textbf{p}$. Then, for each winner user $u_i$, run Algorithm \ref{algorithm:selection} on the set of users except $u_i$. We denote the winner user set in this case as $\mathcal{S'}$ (line 3-4). Finally, each winner user $u_i$ is paid by the maximum virtual bidding price $b_{ik}^v$ that makes her replacing $u_k \in \mathcal{S'}$ as the winner (line 5-7). To achieve this, we should guarantee that
\begin{align*}
\frac{b_{ik}^v - \frac{q_i}{\gamma}}{\sum_{j:\tau_j \in \Gamma_i}\min\{r'_j, 1\}} = \frac{b_k^s + b_k^p\epsilon - \frac{q_k}{\gamma}}{\sum_{j:\tau_j \in \Gamma_k}\min\{r'_k, 1\}}
\end{align*}

Therefore
\begin{align*}
b_{ik}^v = \max\left\{p_i, \frac{\sum\limits_{j:\tau_j \in \Gamma_i} \min\{r'_j, 1\}}{\sum\limits_{j:\tau_j \in \Gamma_k} \min\{r'_j, 1\}} (b_k^s + b_k^p\epsilon - \frac{q_k}{\gamma}) + \frac{q_i}{\gamma}\right\}
\end{align*}

\subsection{Discussions on Practical Implementation}
We emphasize that the fundamental objective of LEPA is not to maintain all \emph{registered users} of the crowdsensing system active in the system, since some of them are unable to execute any sensing tasks in some circumstance, e.g., in the Waze application, some registered users go home and can no longer provide real-time traffic information. Our focus is that, once a registered user arrives at the crowdsensing system (and thus becomes a PU), we do not want she leaves the system for the reason that her expectation is not fulfilled (e.g., her privacy is breached or she is rarely selected). 

In practice, our framework is easily adapted to deal with the situation where the registered users arbitrarily arrive at and leave the crowdsensing system. Specifically, when a registered user arrives at the crowdsensing system, we establish a virtual request queue for this user, and empty the queue when this user leaves. The update of the virtual request queues are the same to Algorithm \ref{algorithm:lyapunov}, and the winner user selection and payment determination in each time slot are the same to Algorithms \ref{algorithm:selection} and \ref{algorithm:price}. Further, our framework is also suitable to deal with the situation where the updating frequency requirements of the tasks in the task pool are heterogeneous. Tasks in the task pool can request for updating any time base on their updating frequency requirements. In each time slot, we should only satisfy that the aggregation accuracy requirements of the tasks that request for updating (rather than all tasks in the task pool) are satisfied. The on-line auction design is the same to Algorithm \ref{algorithm:lyapunov}.

\section{Theoretical Analysis}\label{theoreticalanalysis}
In this section, we provide theoretical analysis to the proposed on-line LPRC auction, including truthfulness, individual rationality, computational complexity and approximation bound.



We prove the truthfulness of the proposed mechanism by the following Lemma:
\begin{lemma}
An auction is truthful if and only if the following two properties hold \cite{jin2015quality}.
\begin{itemize}
  \item \textbf{Monotonicity:} if $u_i$ wins the auction by biding $b_i$ and $y_i$, she also wins by biding $b_i' \leq b_i$ and $\Gamma_i' \supset \Gamma_i$ when other PUs' biddings are fixed. 
  \item \textbf{Critical payment:} PU who wins the auction is paid the maximum bidding price $b_i'$ such that bids $(\Gamma_i, b_i')$ still keeps this PU win. The maximum bidding price $b_i'$ is called critical payment.
\end{itemize}
\end{lemma}

\begin{theorem}[Truthfulness]
The proposed on-line LPRC auction satisfies truthfulness property.
\end{theorem}
\begin{proof}
We prove the truthfulness of the proposed on-line LPRC auction by showing that it satisfies both monotonicity property and critical payment property.
\begin{itemize}
  \item \textbf{Monotonicity:} We select the winner users in the ascending order of their bidding accuracy ratio. Given a fixed $\epsilon$, it is obvious that $u_i$ still wins by biding $\tilde{\Gamma_i} \supset \Gamma_i$ and $\tilde{b_i^s} + \tilde{b_i^p}\epsilon \leq b_i^s + b_i^p\epsilon$.
  \item \textbf{Critical payment:} According to the price determination algorithm proposed in Algorithm \ref{algorithm:price}, the winning PU is paid by the maximum bidding price $b_{ik}^v$.
\end{itemize}
\end{proof}

Notice that, the PU may bid $b_i^s \neq c_i^s$ and $b_i^p \neq c_i^p$ while $b_i^s + b_i^p\epsilon = c_i^s + c_i^p\epsilon$. But this will not increase her payment. Thus, they will still have incentive to bid truthfully.

\begin{theorem}[Individual rationality]
The proposed on-line LPRC auction satisfies individual rationality property.
\end{theorem}
\begin{proof}
For PUs that lose the auction, the utilities are zero according to Definition \ref{definition:pu_utility}. For the winning PUs, they bid the true value $(c_i^s, c_i^p)$ and are paid exactly the maximum bidding prices that still keep them win, guaranteeing that $p_i \geq c_i^s + c_i^p\epsilon$. Thus, $u_i \geq 0,\forall u_i \in \mathcal{U}$, meaning the proposed on-line LPRC auction satisfies individual rationality.
\end{proof}

\begin{theorem}[Computational complexity]
The computational complexity of the proposed on-line LPRC auction is $\mathcal{O}(N^3+N^2M)$.
\end{theorem}
\begin{proof}
Algorithm \ref{algorithm:selection} terminates after $N$ iterations in the worst case. In each iteration, the computation complexity of the winner selection and accuracy requirements update are $\mathcal{O}(N)$ and $\mathcal{O}(M)$, respectively. Therefore, the computation complexity of Algorithm \ref{algorithm:selection} is $\mathcal{O}(N^2+NM)$. Furthermore, Algorithm \ref{algorithm:price} have $N$ iterations of outer loop beside Algorithm \ref{algorithm:price}. Thus, the computation complexity of the proposed on-line LPRC auction is $\mathcal{O}(N^3+N^2M)$
\end{proof}

Then, we analyze the approximation bound of the proposed on-line LPRC auction. Notice that the objective function of the LPRC-OTPM problem comprises of two kinds of variables, i.e., binary variable $\textbf{x(t)}$ and continuous variable $\textbf{p(t)}$, making the analysis of the approximation bound difficult. Further, the optimal solution to the LPRC-OTPM problem is not nonnegative so that it is impractical to achieve a multiplicative approximation bound. To tackle the above challenges, we construct the following LPRC on-line total cost minimization (LPRC-OTCM) problem.
\begin{problemotcm}\label{problem:otcm}
\begin{align*}
\min & \ \ \sum\limits_{i: u_i \in \mathcal{U}} [c_i^s + c_i^p\epsilon - \frac{q_i}{\gamma} + m]x_i \\
s.t. & \ \ \sum\limits_{i: u_i \in \mathcal{U}} y_{ij}x_i \geq r_j, \forall \tau_j \in \mathcal{T} \\
     & \ \ x_i \in \{0,1\},
\end{align*}
\end{problemotcm}
where $m = \max\limits_{i: u_i \in \mathcal{U}} \frac{q_i}{\gamma}$.

Notice that the objective function of the LPRC-OTCM problem contains only the binary variables $\textbf{x(t)}$ and the optimal solution is non-negative. It is easy to come up with a multiplicative bound to the LPRC-OTCM problem. Denote the optimal solution to LPRC-OTCM problem as $M^*$, $\theta = max_{i,j:u_i \in \mathcal{U},\tau_j \in \mathcal{T}}y_{ij}|\Gamma_i|$ and $d = \frac{1}{\Delta r}\sum\limits_{j\in\mathcal{T}}r_j$, where $\Delta r$ is the unit measure of elements in $r_j$. According to \cite{jin2015quality}, we have
\begin{align*}
\sum\limits_{i \in \mathcal{S}} (c_i^s + c_i^p\epsilon - \frac{q_i}{\gamma} + m) \leq 2\theta H_d M^*
\end{align*}
where $H_d = 1 + \frac{1}{2} + \cdots + \frac{1}{d}$.

Then, we derive the relationship between the optimal solution $M^*$ and $P^*$ of the LPRC-OTCM problem and the LPRC-OTPM problem in the following lemma.
\begin{lemma}\label{lemma:p2m}
The relationship between $M^*$ and $P^*$ is given by
\begin{align*}
M^* \leq P^* + mn.
\end{align*}
where $n$ is the total number of PUs.
\end{lemma}
The proof can be found in Appendix \ref{appendix:p2m}. Define $\delta = (\max\limits_{i \in \mathcal{U}} (b_{k_i}^s + b_{k_i}^p\epsilon - \frac{q_{k_i}}{\gamma}) / 
(\min\limits_{i \in \mathcal{U}} (b_{k_i}^s + b_{k_i}^p\epsilon - \frac{q_{k_i}}{\gamma} + m))$. The approximation bound of the proposed on-line LPRC auction is given by Theorem \ref{theorem:ratio}.
\begin{theorem}[Approximation Bound]\label{theorem:ratio}
The approximation bound of the proposed on-line LPRC auction is given by
\begin{align*}
P \leq 2\delta \theta dH_d(P^* + mn)
\end{align*}
\end{theorem}
where $P$ is the solution achieved by the proposed on-line LPRC auction. The proof can be found in Appendix \ref{appendix:ratio}.

\section{Performance Evaluation}\label{simulation}
In this section, we first introduce two baselines and the simulation settings. Then, we present the simulation results.

\subsection{Simulation Setup}
Firstly, we introduce two baselines to illustrate the effectiveness of the proposed on-line LPRC auction. The first baseline is \emph{static auction} without long-term participation constraints, FC only aim to minimize its total payment at \emph{current slot}. In the static auction, FC chooses the winner users by their accuracy bidding ratio and pay them with their critical payments at each time slot. The other baseline is an LPRC auction with compulsory long-term participation constraints, i.e., compulsory LPRC auction. The compulsory LPRC auction chooses winner users by their accuracy bidding ratio meanwhile guarantees that all PUs will be selected at least once every $\frac{1}{D}$ time slots (we set $D = 0.2$ in the following simulations). The winner users are paid their critical payments.
\begin{table}[!htp]
	\caption{Simulation settings}
	\centering
	\tabcolsep=2pt
	\begin{tabular}[c]{|c|c|c|c|c|c|c|c|}
		\hline \label{table:simulationsettings}
		Setting 	& $\alpha_j$&	$\delta_j$	&	$c_i^s,c_i^p$	&	$|\Gamma_i|$   & $\epsilon$ &	   n 	   &   k   	\\ \hline\hline
		
		I			& $[1,2]$	&  $[0.1,0.2]$  &     $[1,2]$		&	$[5,10]$	   & 	$1$	    &    $100$	   &  $10$	\\ \hline

		II			& $[1,2]$	&  $[0.1,0.2]$  &     $[1,2]$		&	$[5,10]$	   & 	$1$	    & $[100,200]$  &  $10$	\\ \hline
		
		III 		& $[1,2]$	&  $[0.1,0.2]$	&     $[1,2]$		&	$[5,10]$       &  $[0.5,2]$	&    $100$     &  $10$	\\ \hline
	\end{tabular}
\end{table}

Then, we provide the simulation settings in Table \ref{table:simulationsettings}. In setting I, II and III, $\tau_j$'s accuracy requirement $\alpha_j,\delta_j$ and $u_i$'s cost $c_i^s,c_i^p$ take values from the interval shown in Table \ref{table:simulationsettings}, uniformly. $|\Gamma_i|$ stands for the number of tasks that $u_i$ can execute and takes values from the interval $[5,10]$ uniformly. We choose $|\Gamma_i|$ tasks for $u_i$ uniformly from the task set $\mathcal{T}$. In setting I, we fix the PPL $\epsilon$, the number of PUs $n$ and the number of tasks $k$ and show the advantage of the proposed on-line LPRC auction. In setting II, we fix $\epsilon$ and $k$, and demonstrate the impact of the number of PUs on FC's average total payment. In setting III, we fix $n$ and $k$, and illustrate the impact of PUs' PPL $\epsilon$ on FC's average total payment.

\subsection{Simulation Results}
\begin{figure}[ht]
\setlength{\abovecaptionskip}{0pt} 
\setlength{\belowcaptionskip}{-5pt}
\begin{center}
\includegraphics[width=0.3\textwidth]{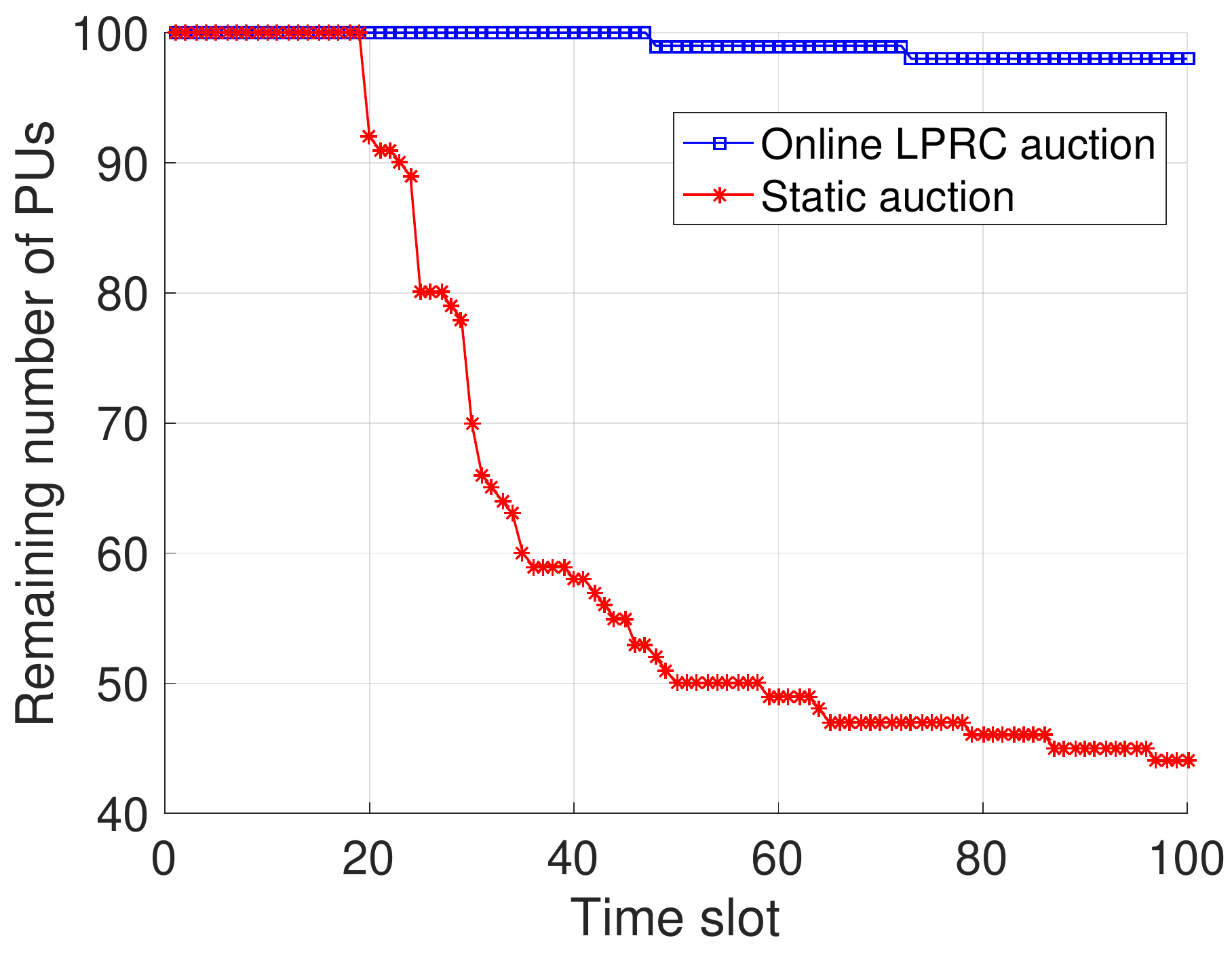}
\end{center}
\caption{The remaining number of PUs Vs. time slot.}
\label{fig:pu_remain}
\end{figure}

\begin{figure}[ht]
\setlength{\abovecaptionskip}{0pt} 
\setlength{\belowcaptionskip}{-5pt}
\begin{center}
\includegraphics[width=0.3\textwidth]{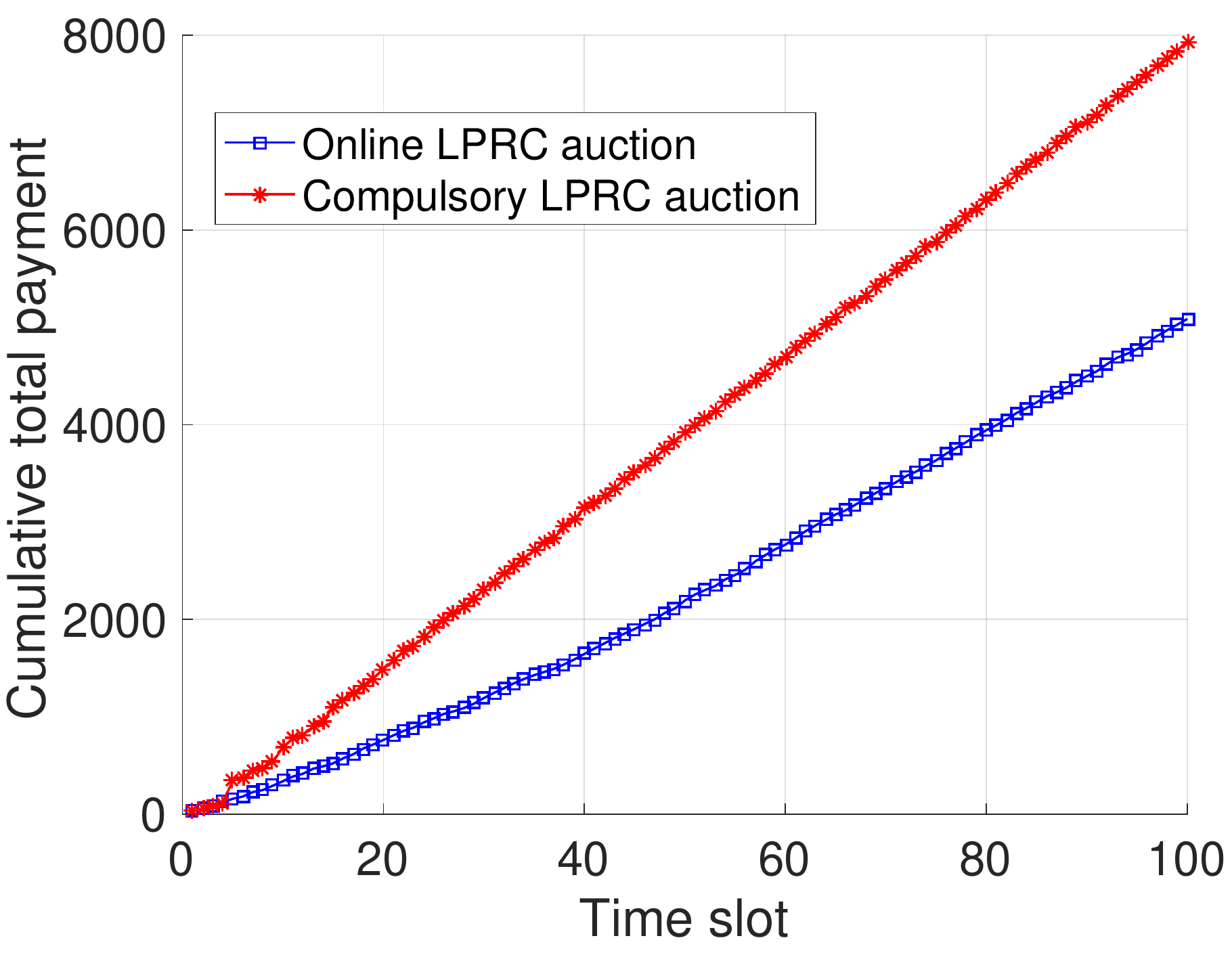}
\end{center}
\caption{FC's total payment Vs. time slot.}
\label{fig:lyapunov2naive}
\end{figure}

\begin{figure}[ht]
\setlength{\abovecaptionskip}{0pt} 
\setlength{\belowcaptionskip}{-5pt}
\begin{center}
\includegraphics[width=0.3\textwidth]{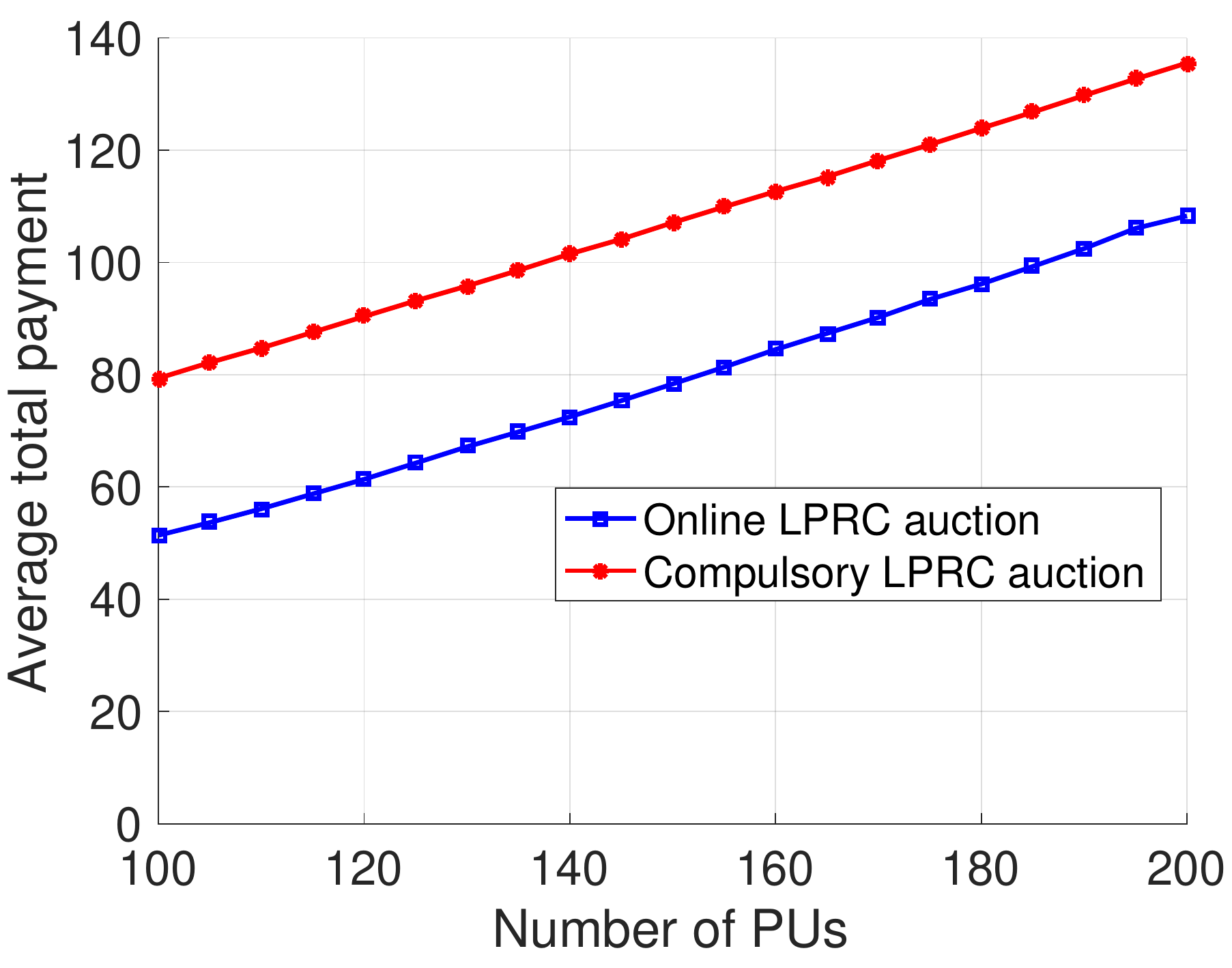}
\end{center}
\caption{Number of PUs Vs. average total payment of FC.}
\label{fig:impact_n_avepayment}
\end{figure}

\begin{figure}[ht]
\setlength{\abovecaptionskip}{0pt} 
\setlength{\belowcaptionskip}{-5pt}
\begin{center}
\includegraphics[width=0.3\textwidth]{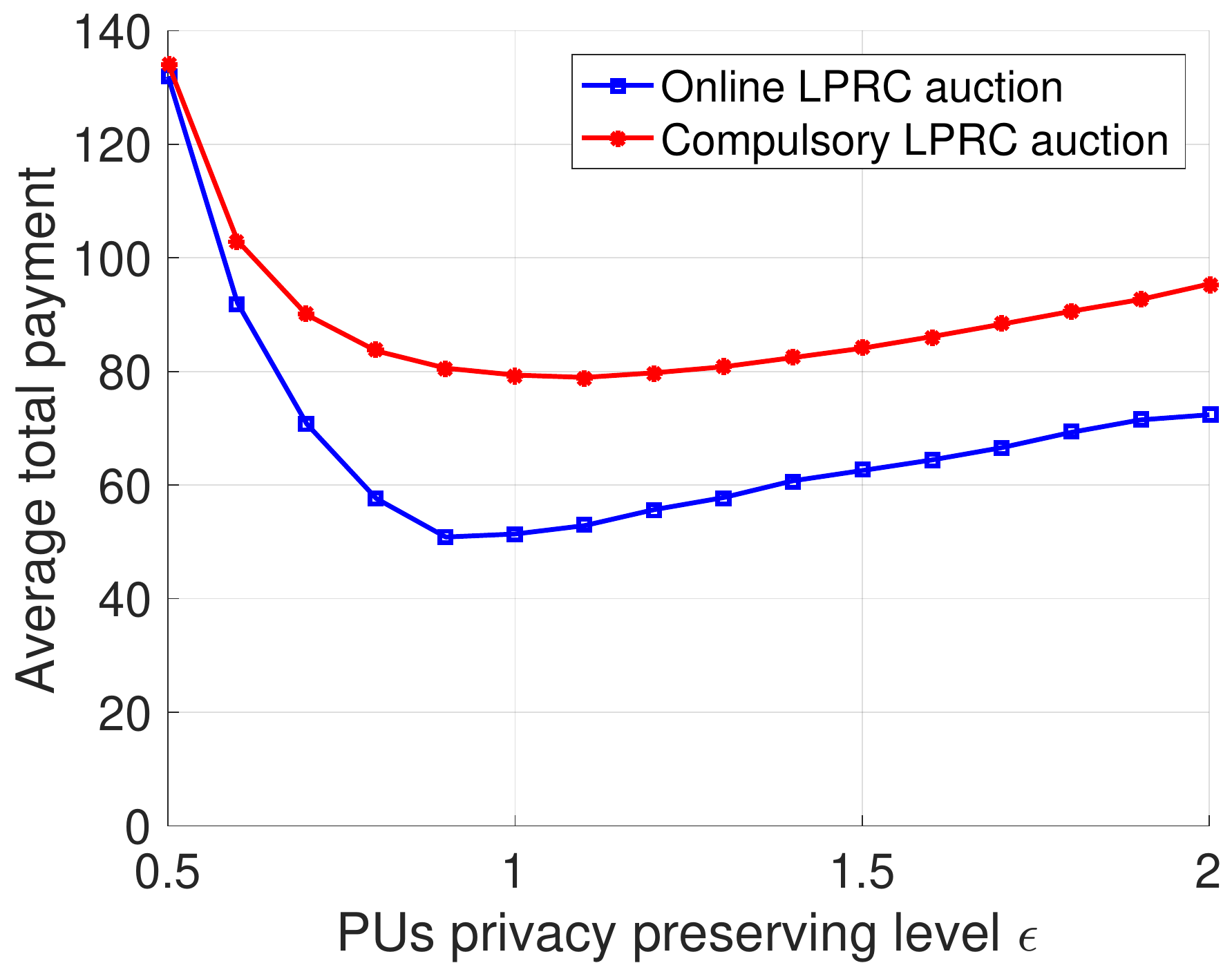}
\end{center}
\caption{PUs' PPLs $\epsilon$ Vs. average total payment of FC.}
\label{fig:impact_epsilon_avepayment}
\end{figure}

Figure \ref{fig:pu_remain} shows the remaining number of PUs at each time slot in setting I. In this simulation, we assume that a PU would drop out of the system if she was not selected for $20$ consecutive time slots. Since in this situation, PUs in the compulsory LPRC auction will not leave the system, we only compare the on-line LPRC auction with the static auction in this simulation. We observe that, in the on-line LPRC auction, the number of PUs remains almost unchanged. However, the remaining number of PUs drops dramatically after $20$ time slots in the static auction. After $100$ time slots, less than half of PUs still stay in the system. The simulation result shows that the proposed on-line LPRC auction can efficiently maintain the long-term participation of PUs.

In Figure \ref{fig:lyapunov2naive}, we compare the cumulative total payment of the on-line LPRC auction with the compulsory LPRC auction at each time slot in setting I. Figure \ref{fig:lyapunov2naive} shows that the proposed on-line LPRC auction can efficiently reduce FC's cumulative total payment. The reason is that, the proposed on-line LPRC auction has jointly optimized FC's payment and PUs' long-term participation requirement, whereas the compulsory LPRC auction strictly guarantees PUs' long-term participation requirements and ignores the optimality of FC's payment.

Figure \ref{fig:impact_n_avepayment} shows the impact of the number of PUs on FC's average total payment in setting II. We observe that, the average total payment increases when the number of PUs increases. Because more PUs means more virtual sensing request queues to stabilize, affecting the optimality of FC's average total payment.

In Figure \ref{fig:impact_epsilon_avepayment}, we show the impact of PUs' PPL $\epsilon$ on FC's average total payment in setting III. Figure \ref{fig:impact_epsilon_avepayment} shows that when $\epsilon$ is small, the average total payment decreases when $\epsilon$ increases. Since larger $\epsilon$ means higher data quality, FC can recruit less PUs to execute each task, leading to the decrease of the average total payment. However, when $\epsilon$ continues to increase, PUs' privacy loss increases dramatically, resulting in the increase of the payment to a single PU. As a result, the average total payment increases. The simulation result demonstrates that, there exists an optimal $\epsilon$ to minimize FC's average total payment, FC can determine the optimal $\epsilon$ with numerical calculation as this simulation.

\section{Conclusion}\label{conclusion}
In this paper, we designed a long-term privacy-preserving incentive mechanism LEPA for real-time data aggregation. LEPA takes both long-term participation constraints and PUs' privacy into consideration, which are essential to incentivize participation in real-time data aggregation. Due to PUs' strategic behavior and combinatorial nature of the tasks, we proposed a computationally efficient mechanism with near-to-optimal performance to jointly optimize FC's payment and PUs' participation. Furthermore, we ensure that the proposed on-line auction satisfies other desirable properties, including truthfulness and individual rationality. The effectiveness of the proposed on-line auction is validated by both theoretical analysis and extensive simulations.

\bibliographystyle{IEEEtran}
\bibliography{cited}

\appendix
\section{Proof of Lemma 1}\label{appendix:relationship} 
\begin{proof}
We utilize $\hat{s}_j$ to denote the privacy-preserving aggregation of task $\tau_j$. Thus, the aggregation error between privacy-preserving data and real sensing data is given by
\begin{align*}
\hat{s}_j - s_j &= \frac{1}{|\mathcal{S}|}\sum_{i: u_i \in \mathcal{S}}(d_i+\eta_i) 
- \frac{1}{|\mathcal{S}|}\sum_{i: u_i \in \mathcal{S}}d_i
&=\frac{1}{|\mathcal{S}|}\sum_{i: u_i \in \mathcal{S}}\eta_i.
\end{align*}
where $|\mathcal{S}|$ is the cardinality of winner user set $\mathcal{S}$, $d_i$ and $\eta_i$ are the real sensing data and noise added to $d_i$, respectively.

Recall that the variance of the Laplacian random variable $\eta_i\sim Lap(0,a_i)$ is $2a_i^2$, i.e., $D(\eta_i)=2a_i^2$. Thus, for independent Laplacian variables, we have
\begin{align*}
D(\frac{1}{|\mathcal{S}|}\sum_{i: u_i \in \mathcal{S}}\eta_i)=\frac{2}{|\mathcal{S}|^2}\sum_{i: u_i \in \mathcal{S}}a_i^2.
\end{align*}

From the \emph{Chebyshev's inequality}, we derive that
\begin{align*}
\mathbb{P}[|\hat{s}_j - s_j| \geq \alpha_j] \leq \frac{2}{\alpha_j^2 |\mathcal{S}|^2}\sum_{i: u_i \in \mathcal{S}} a_i^2.
\end{align*}

To achieve $(\alpha_j, \delta_j)$-accuracy for task $\tau_j$, we should guarantee that
\begin{align*}
\frac{2}{\alpha_j^2 |\mathcal{S}|^2}\sum_{i: u_i \in \mathcal{S}} a_i^2 \leq \delta_j.
\end{align*}

Therefore, in time slot $t$, substituting $b_i=\frac{\gamma}{\epsilon_i}$ into the above formula and do some algebra, we derive Lemma \ref{lemma:relationship}.
\end{proof}

\section{Proof of Lemma 6}\label{appendix:p2m}
\begin{proof}
Assume $(x^*, p^*)$ is the optimal solution to the LPRC-OTPM problem, it is obvious that $P^* = \sum_{i \in \mathcal{U}} (p_i^* - \frac{q_i}{\gamma})$. Further, since the proposed mechanism is truthful and individual rational, we have $p^* \geq c_i^s + c_i^p\epsilon$.

Notice that the constraints of the LPRC-OTPM problem and the LPRC-OTCM problem are the same, $(x^*, p^*)$ is also feasible to the LPRC-OTCM problem. Thus
\begin{align*}
M^* & \leq \sum\limits_{i \in \mathcal{U}} (c_i^s + c_i^p\epsilon - \frac{q_i}{\gamma} + m)x_i^* \\
	& \leq \sum\limits_{i \in \mathcal{U}} (p_i^* - \frac{q_i}{\gamma} + m)x_i^* \\
	& = P^* + \sum\limits_{i \in \mathcal{U}} mx_i^* 
	 \leq P^* + mn
\end{align*}
which concludes the proof.
\end{proof}

\section{Proof of Theorem 7}\label{appendix:ratio}
\begin{proof}
Notice that
\begin{align*}
&\sum\limits_{i \in \mathcal{S}} (c_i^s + c_i^p\epsilon - \frac{q_i}{\gamma} + m) 
\geq |\mathcal{S}| \min\limits_{i \in \mathcal{U}} \{c_i^s + c_i^p\epsilon - \frac{q_i}{\gamma} + m\}
\end{align*}

From Algorithm \ref{algorithm:price}, we know that for every $u_i$, there exists an $u_{k_i}$ such that
\begin{align*}
p_i = (b_{k_i}^s + b_{k_i}^p\epsilon - \frac{q_{k_i}}{\gamma})
\frac{\sum\limits_{j:\tau_j \in \Gamma_i} \min\{r'_j, 1\}}{\sum\limits_{j:\tau_j \in \Gamma_{k_i}} \min\{r'_j, 1\}} + \frac{q_i}{\gamma}
\end{align*}

Therefore, we have
\begin{align*}
 \sum\limits_{i \in \mathcal{S}} p_i - \frac{q_i}{\gamma} 
= 	 & (b_{k_i}^s + b_{k_i}^p\epsilon - \frac{q_{k_i}}{\gamma})
\frac{\sum\limits_{j:\tau_j \in \Gamma_i} \min\{r'_j, 1\}}{\sum\limits_{j:\tau_j \in \Gamma_{k_i}} \min\{r'_j, 1\}} \\
\leq & d|\mathcal{S}|\max\limits_{i \in \mathcal{U}} (b_{k_i}^s + b_{k_i}^p\epsilon - \frac{q_{k_i}}{\gamma}) \\
\leq & d\frac{\max\limits_{i \in \mathcal{U}} (b_{k_i}^s + b_{k_i}^p\epsilon - \frac{q_{k_i}}{\gamma})}
{\min\limits_{i \in \mathcal{U}} (b_{k_i}^s + b_{k_i}^p\epsilon - \frac{q_{k_i}}{\gamma} + m)}
\sum\limits_{i \in \mathcal{S}} (c_i^s + c_i^p\epsilon - \frac{q_i}{\gamma} + m) \\
=    & d\delta \sum\limits_{i \in \mathcal{S}} (c_i^s + c_i^p\epsilon - \frac{q_i}{\gamma} + m) \\
\leq & 2\theta\delta dH_d M^* \\
\leq & 2\theta\delta dH_d(P^* + mn)
\end{align*}
which concludes the proof.
\end{proof}

\end{document}